\documentclass{article}
\usepackage[utf8]{inputenc}
\usepackage{amsthm,amsmath,amssymb}
\usepackage{subfig}
\usepackage{xspace}
\usepackage{graphicx}
\usepackage{multirow}

\newtheorem{lemma}{Lemma}
\newtheorem{definition}{Definition}
\newtheorem{corollary}{Corollary}
\newtheorem{theorem}{Theorem}

\newcommand{\reddit}[1]{\texttt{#1}}

\newcommand{\adv}{\ensuremath{\mathsf{Adv}}}

\makeatletter
\newcommand\delaytext@errorcmd{\@latex@error}
\newcommand\delaytext[2]{%
  \long\expandafter\gdef\csname delaytext:#1\endcsname{#2}%
  \AtEndDocument{%
    \expandafter\ifx\csname delaytext-used:#1\endcsname\relax
    \delaytext@errorcmd{delaytext #1 unused. Use \noexpand\usedelayedtext{#1}
      somewhere.}\fi}%
}
\newcommand\usedelayedtext[1]{%
  \dropdelayedtext{#1}%
  \expandafter\ifx\csname delaytext:#1\endcsname\relax
  \delaytext@errorcmd{delaytext #1 undefined.}\fi
  \@nameuse{delaytext:#1}%
}
\newcommand\dropdelayedtext[1]{%
  \expandafter\ifx\csname delaytext-used:#1\endcsname\relax\else
  \delaytext@errorcmd{delaytext #1 used twice.}\fi
  \global\@namedef{delaytext-used:#1}{1}%
  \expandafter\ifx\csname delaytext:#1\endcsname\relax
  \@latex@error{delaytext #1 undefined}\fi
}
\newcommand\delaytextwarnonly{%
  \let\delaytext@errorcmd\@latex@warning
}

\newcommand{\mathcmd}[1]{\ensuremath{#1}\xspace}

\newcommand\tableofcontentsapp{%
   { 
   }
   {
     \@starttoc{mytoc}%
     \bigskip\bigskip
   }
}

\newcommand{\addsubsection}[2]{
  \addtocontents{mytoc}{\iftrue\vskip 0.1em\fi}
  \addcontentsline{mytoc}{subsection}{\ref{#1} \quad #2}
}

\newcounter{myExampleCounter}
\refstepcounter{myExampleCounter}

\newenvironment{framedFigure}[2]
{
\begin{figure}[t]
\def\tmpCaption{#1}
\def\tmpLabel{#2}
\begin{framed}
}
{
\vspace{-0.5em}
\end{framed}
\vspace{-0.5em}
\caption{\tmpCaption}\label{\tmpLabel}
\end{figure}
}

\newcommand{\mathstring}[1]{\mathcmd{\mathsf{#1}}}

\newcommand{\Adv}{\mathcmd{\mathstring{Adv}}}

\renewcommand{\paragraph}[1]{\smallskip\noindent\textbf{#1}\hspace{1ex}}

\newcommand{\Obs}{\mathcmd{\mathcal{O}}}

\newcommand{\vocabulary}{\mathcmd{\mathcal{V}}}

\newcommand{\attribute}{\mathcmd{\alpha}}
\newcommand{\Attributes}{\mathcmd{\mathcal{A}}}

\newcommand{\entity}{\mathcmd{\epsilon}}
\newcommand{\Entities}{\mathcmd{\mathcal{E}}}
\newcommand{\Know}{\mathcmd{\kappa}}

\newcommand{\Profiles}{\mathcmd{\mathcal{P}}}

\newcommand{\stringCount}{\mathcmd{\mathsf{count}}}

\renewcommand{\div}{\mathcmd{D}}

\newcommand{\dist}{\mathcmd{\mathsf{dist}}}

\renewcommand{\adv}{\mathcmd{\mathsf{Adv}}}

\newcommand{\Policy}{\mathcmd{\mathcal{R}}}
\newcommand{\policy}{r}

\newcommand{\dom}{\mathcmd{\mathsf{dom}}}

\title{From Closed-world Enforcement to Open-world Assessment of Privacy}
\date{}

\author{Michael Backes$^{1,2}$ \hspace{3em} Pascal Berrang$^1$ \hspace{3em} Praveen Manoharan$^1$\\
~\\
~$^1$CISPA, Saarland University\quad ~$^2$ MPI-SWS\\
~\\
backes@cs.uni-saarland.de\\
berrang@cs.uni-saarland.de\\
manoharan@cs.uni-saarland.de
}

\begin{document}
\pagestyle{plain}

\maketitle
\begin{abstract}
In this paper, we develop a \emph{user-centric privacy framework for quantitatively assessing the exposure} of personal information 
in open settings. 
Our formalization addresses key-challenges posed by such open settings, such as the unstructured dissemination of heterogeneous information and the necessity of user- and context-dependent privacy requirements. We propose a new definition of information sensitivity derived from our formalization of privacy requirements, and, as a sanity check, show that hard non-disclosure guarantees are impossible to achieve in open settings. 

After that, we provide an instantiation of our framework to address the \emph{identity disclosure} problem, leading to the 
novel notion of $d$-convergence.
$d$-convergence is based  on indistinguishability of entities and it bounds the likelihood with which an adversary successfully links two profiles of 
the same user across online communities.

Finally, we provide a \emph{large-scale evaluation} of our framework on a collection of 15 million comments collected from the Online Social Network Reddit. 
Our evaluation validates the notion of $d$-convergence for assessing the linkability of entities in our 
data set and provides deeper insights into the data set's structure.
\end{abstract}
\newpage
\setcounter{tocdepth}{2}
\makeatletter	
\makeatletter
\tableofcontents
\newpage

\section{Introduction}\label{section:introduction}
The Internet has undergone dramatic changes in the last two decades,
evolving from a mere communication network to a global multimedia
platform in which billions of users not only actively exchange
information, but increasingly conduct sizable parts of their daily lives. 
While this transformation has brought tremendous benefits
to society, it has also created new threats to online privacy that
existing technology is failing to keep pace with. 
Users tend to reveal personal
information without considering the widespread, easy accessibility,
potential linkage and permanent nature of online data. Many cases
reported in the press show the resulting risks, which range from
public embarrassment and loss of prospective opportunities (e.g., when
applying for jobs or insurance), to personal safety and property risks
(e.g., when sexual offenders or burglars learn users' whereabouts online).
The resulting privacy awareness and privacy concerns of Internet users have been further amplified
by the advent of the Big-Data paradigm and the aligned business models
of personalized tracking and monetizing personal information in an unprecedented manner.

Developing a suitable methodology to reason about 
the privacy of users in such a large-scale, open web setting, as well 
as corresponding tool support in the next step, requires at its core a formal privacy model 
that lives up to the now increasingly dynamic dissemination of unstructured, heterogeneous user content on the Internet:
While users traditionally shared information mostly using public profiles with static information about themselves, nowadays they
disseminate personal information in an unstructured, highly dynamic manner, through content they 
create and share (such as blog entries, user comments,  a ``Like'' on Facebook), or through the people 
they befriend 
or follow. 
Furthermore, ubiquitously available background knowledge about a dedicated user needs to be appropriately reflected within the
model and its reasoning tasks, as it can decrease a user's privacy by inferring further sensitive information. 
As an example, Machine Learning and other Information Retrieval techniques provide comprehensive approaches for profiling a user's actions
across multiple Online Social Networks, up to a unique identification of a given user's profiles for each such network. 


Prior research on privacy has traditionally focused on closed database settings -- characterized by a complete view on 
structured data and a clear distinction of key- and sensitive attributes --
and has aimed for  strong privacy guarantees using global data sanitization.
These approaches, however, are inherently inadequate if such closed settings are replaced by open settings as described above, where 
unstructured and heterogeneous data is being disseminated, where individuals have a partial 
view of the available information, and where global data sanitization is impossible and hence 
strong guarantees have to be replaced by probabilistic privacy assessments.

As of now, \emph{even the basic methodology is missing} for offering users
technical means to comprehensively
assess the privacy risks incurred by their data dissemination, and their daily online activities in general. Existing privacy models such as $k$-anonymity~\cite{Sweeney02}, $l$-diversity~\cite{Machanavajjhala07}, $t$-closeness~\cite{Li07} and 
the currently most popular notion of Differential Privacy~\cite{Dwork08} follow a database-centric approach that is inadequate to
meet the requirements outlined above. We refer the reader to Section~\ref{section:general_idea} for further discussions on existing privacy models.


\subsection{Contribution}

In this paper, we present a rigorous methodology for quantitatively assessing the exposure of personal information 
in open settings. Concretely, the paper makes the following
three tangible contributions: 
(1) a formal framework for reasoning about the disclosure of personal information in open settings, 
(2) an instantiation of the framework for reasoning about the identity disclosure problem, 
and (3) an evaluation of the framework on a collection of $15$ million comments collected 
from the Online Social Network Reddit.

\medskip
\noindent
\emph{A Formal Framework for Privacy in Open Settings.} 
We propose a novel framework for addressing the essential challenges of privacy in open settings, such as providing 
a data model that is suited for dealing with unstructured dissemination of heterogeneous
information through various different sources and a flexible definition of user-specific 
privacy requirements that allow for the specification of context-dependent privacy goals.
In contrast to most existing approaches, our framework strives to assess the degree of 
exposure individuals face, in contrast to trying to enforce an individual's privacy requirements. 
Moreover, our framework technically does not differentiate 
between non-sensitive and sensitive attributes a-priori, but rather starts 
from the assumption that all data is equally important and can lead to privacy risks.
More specifically, our model captures
the fact that the sensitivity of attributes 
is highly user- and context-dependent by 
deriving information sensitivity from each user's 
privacy requirements. As a sanity check we prove that hard
non-disclosure guarantees cannot be provided for the open setting
in general, providing incentive for novel approaches for assessing 
privacy risks in the open settings.

\medskip
\noindent
\emph{Reasoning about Identity Disclosure in Open Settings.} 
We then instantiate our general privacy framework for the specific use case of identity disclosure. Our framework
 defines and assesses identity disclosure (i.e., identifiability and linkability of identities) 
by utilizing entity similarity, i.e., an entity is private in a collection of 
entities if it is sufficiently similar to its peers.
At the technical core of our model is the new notion of $d$-convergence, which quantifies the similarity 
of entities within a larger group of entities. It hence provides the formal grounds to assess the ability of any single 
entity to blend into the crowd, i.e., to hide amongst peers.
The $d$-convergence model is furthermore capable of assessing identity disclosure risks specifically for single entities. 
To this end, we extend the notion of $d$-convergence to the novel notion of $(k,d)$-anonymity, which allows 
for entity-centric identity disclosure risk assessments by requiring $d$-convergence in the local neighborhood of a given entity. 
Intuitively, this new notion provides a generalization of $k$-anonymity 
that is not bound to matching identities based on pre-defined key-identifiers.

\medskip
\noindent
\emph{Empirical Evaluation on Reddit.} 
Third, we perform an instantiation of our identity disclosure model for the important use case of analyzing user-generated text content in order to characterize specific user profiles. We use unigram frequencies extracted from user-generated content as user attributes, and we subsequently demonstrate that the
resulting unigram model can indeed be used for quantifying the degree of anonymity of -- and ultimately, for differentiating -- individual entities.
For the sake of exposition, we apply this unigram model to a collection of $15$ million comments collected from the Online Social Network Reddit.
The computations were performed on two Dell PowerEdge R820 with 64 virtual cores each at $2.60$GHz over the course of six weeks. Our evaluation shows that $(k,d)$-anonymity suitably assesses an identity's anonymity and provides deeper insights into the data set's structure.

\subsection{Outline}
We begin by discussing related work in Section~\ref{section:related_work} and explain why existing
privacy notions are inadequate for reasoning about privacy in open web settings in Section~\ref{section:open_settings}. 
We then define our privacy framework in Section~\ref{section:general_model} and instantiate it for reasoning about 
identity disclosure in Section~\ref{section:linkability_model}. In Section~\ref{section:experiment} we perform a basic 
evaluation of the identity disclosure model on the Reddit Online Social Network. We summarize our findings~\ref{section:conclusion}.
\section{Related Work}\label{section:related_work}
In this section, we give an overview over other relevant related work that has not yet been considered in the previous subsection.

\paragraph{Privacy in Closed-world Settings.}
The notion of privacy has been exhaustively discussed for specific settings such as statistical databases, as well as for more general settings. Since we 
already discussed the notions of $k$-anonymity~\cite{Sweeney02}, $l$-diversity~\cite{Machanavajjhala07} $t$-closeness~\cite{Li07} and 
Differential Privacy~\cite{Dwork08} in Section~\ref{section:general_idea} in great detail, we will now discuss further such notions.

A major point of criticism of Differential Privacy, but also the other existing privacy notions, 
found in the literature~\cite{Kifer11} is the (often unclear) trade-off between utility and privacy that is incurred by applying database sanitation techniques to achieve privacy. Several works 
have shown that protection against attribute disclosure cannot be provided in settings that consider an adversary with arbitrary 
auxiliary information~\cite{Dinur03,Dwork06,Dwork08-2}. We later show, as sanity check, that in our formalization of privacy in open settings, general non-disclosure guarantees are indeed impossible to achieve. By providing the necessary formal groundwork in this paper, we hope to stimulate research on \emph{assessing} privacy risks in open settings, against explicitly spelled-out adversary models.  

Kasiviswanathan and Smith~\cite{Semantic_of_Differential_Privacy} define the notion of $\epsilon$-semantic privacy to capture general non-disclosure guarantees. We define our adversary model in a similar fashion as in their formalization and we use $\epsilon$-semantic privacy to show that general non-disclosure guarantees cannot be meaningfully provided in open settings.

Several extensions of the above privacy notions have been proposed in the literature to provide privacy guarantees in use cases that differ from traditional database privacy~\cite{AnoA,Chatzikokolakis13,Zheleva09,Heatherly13,Zhou11,Chen14}. These works aim at suitably transforming different settings into a 
database-like setting that can be analyzed using differential privacy. Such a transformation, however, often abstracts away from essential components of these settings, and as a result achieve impractical privacy guarantees. As explained in Section~\ref{section:general_idea}, the open web setting is particularly ill-suited for such transformations.  

Specifically for the use case in Online Social Networks (in short, OSNs), many works~\cite{Zheleva09,Kosinski14,Heatherly13,Zhou11,Chen14} apply the existing database 
privacy notions for reasoning about attribute disclosure in OSN data. These works generally impose a specific structure on OSN data, 
such as a social link graph, and reason about the disclosure of private attributes through this structure.  
Zhaleva et al.~\cite{Zheleva09} show that mixed public and private profiles do not necessarily protect the private part of a profile since they can be inferred from the public part.
Heatherly et al.~\cite{Heatherly13} show how machine learning techniques can be used to infer private information from publicly available information.
Kosinksi et al.~\cite{Kosinski14} moreover show that machine learning techniques can indeed be used to predict personality traits of users and their online behavior.
Zhou et al.~\cite{Zhou11} apply the notions of $k$-anonymity and $l$-diversity to data protection in OSNs and discuss the complexity of finding private subsets. Their approach does however suffer from the same problems these techniques have in 
traditional statistical data disclosure, where an adversary with auxiliary information can easily infer information about any specific user.
Chen et al.~\cite{Chen14} provide a variation of differential privacy which allows for privacy and protection against edge-disclosure attacks in the correlated setting of OSNs.
The setting, however, remains static, and it is assumed that the data can be globally sanitized in order to provide protection against attribute 
disclosure. Again, as discussed in Section~\ref{section:general_idea}, this does not apply to the open web setting with its highly unstructured dissemination of data.

\paragraph{Statistical Language Models.}
Statistical Language Models for information retrieval have first been introduced by Ponte and Croft~\cite{Ponte98} as an alternative approach for document retrieval and are inspired by language models for Speech Recognition and Natural Language Processing~\cite{Song99,Rosenfeld00}. 
They have subsequently been focus of a long line of research (examples include ~\cite{Lavrenko01,Zhai04,Uzuner05,Zhai08}) that further develop the basic statistical language model approach and its benefits. While Statistical Language Models have not been shown to perform better than other established retrieval methods~\cite{Zhai08}, we found that the Statistical Language Model formulation is closer than other options to what we require in expressing and solving indistinguishability problems that arise in computer security. 
\section{Privacy in Open Settings}\label{section:open_settings}
Before we delve into the technical parts of this paper, we give an informal overview over privacy in the Internet of the future. To this end, we first provide an example that illustrates some of
the aspects of privacy in the Internet, and then in detail discuss the challenges of privacy in the Internet and why existing privacy notions are not applicable to this setting.
\subsection{Example}
Consider the following example: Employer Alice receives an application by potential employee Bob which contains personal information about Bob. Before she makes the decision on the employment of Bob, however, she searches the <internet and tries to learn even more about her potential employee. A prime source of information are, for example, Online Social Networks (OSNs) which Alice can browse through. If she manages to identify Bob's profile in such an OSN she can then learn more about Bob by examining the publicly available information of this profile.

In order to correctly identify Bob's profile in an OSN, Alice takes the following approach: 
based on the information found in Bob's application, she constructs a model $\theta_B$ that contains all attributes, such as 
name, education or job history, extracted from Bob's application. She then compares this model $\theta_B$ to the profiles 
$P_1,\ldots, P_n$ found in the OSNs and ranks them by similarity to the model $\theta_B$. Profiles that 
show sufficient similarity to the model $\theta_B$ are then chosen by Alice as belonging to Bob. After identifying the (for 
Alice) correctly matching profiles $P_1^\ast,\ldots, P_i^\ast$ of Bob, Alice can finally merge their models 
$\theta_1^{\ast},\ldots,\theta_i^{\ast}$ with $\theta_B$ to increase her knowledge about Bob.

Bob now faces the problem that Alice could learn information about him that he does not want her to learn. He basically has 
two options: he either does not share this critical information at all, or makes sure that his profile is not identifiable 
as his. In OSNs such as Facebook, where users are required to identify themselves, Bob can only use the first option.
In anonymous or pseudonymous OSNs such as Reddit or Twitter, however, he can make use of the second option. He then has to 
make sure that he does not share enough information on his pseudonymous profiles that would allow Alice to link his 
pseudonymous profile to him personally.

Privacy in the open web is mostly concerned with the second option: we cannot protect an entity \entity against sharing personal 
information through a profile which is already uniquely identified with the entity \entity. We can, however, estimate how well an 
pseudonymous account of \entity can be linked to \entity, and through this link, learn personal information about \entity.
As the example above shows, we can essentially measure privacy in terms of similarity of an entity $\entity$ 
in a collection of entities $\Entities$. 

The identifiability of $\entity$ then substantially depends on the 
attributes $\entity$ exhibits in the context of $\Entities$ 
and does not necessarily follow the concept of personally 
identifiable information (PII) as known in the more common understanding of privacy and in privacy and data-protection 
legislation~\cite{EUPII}: here, privacy protection only goes as far as protecting this so-called personally identifiable 
information, which often is either not exactly defined, or restricted to an a-priori-defined set of attributes such as name, 
Social Security number, etc. We, along with other authors in the literature~\cite{Narayanan09,Narayanan10}, find however that the set of 
critical attributes that need to be protected differ from entity to entity, and from community to community. For example, in 
a community in which all entities have the name ``Bob'', exposing your name does not expose any information about yourself. 
In a different community, however, where everyone has a different name, exposing your name exposes a lot of information 
about yourself.

In terms of the privacy taxonomy formulated by Zheleva and Getoor~\cite{Zheleva11}, the problem we face corresponds to the identity disclosure problem, where one tries to identify whether and how an identity is represented in an OSN. We think that this is one of the main concerns of users of frequently used OSNs, in particular those that allow for pseudonymous interactions: users are able to freely express their opinions in these environments, assuming that their opinions cannot be connected to their real identity. However, any piece of information they share in their interactions can leak personal information that can lead to identity disclosure, defeating the purpose of such pseudonymous services.

To successfully reason about the potential disclosure of sensitive information in such open settings, we first have to consider various challenges that
have not been considered in traditional privacy research. After presenting these challenges, we discuss the implications of these challenges on some of the existing privacy notions, before we consider other relevant related work in the field.

\subsection{Challenges of Privacy in Open Settings}\label{section:challenges}
In this subsection, we introduce the challenges induced by talking about privacy in open settings:

\medskip
\noindent
\emph{C1) Modeling heterogeneous information.} We require an information model that allows for modeling various types of information and that reflects the heterogeneous information shared throughout the Internet. This models needs to adequately represent personal information that can be inferred from various sources, such as static profile information or from user-generated content, and should allow statistical
assessments about the user, as is usually provided by knowledge inference engines. We propose a solution to this challenge in Section~\ref{section:data_model}.

\medskip
\noindent
\emph{C2) User-specified privacy requirements.} We have to be able to formalize user-specified privacy requirements. This formalization should use 
the previously mentioned information model to be able to cope with heterogeneous information, and specify which information should be protected from being publicly disseminated. We present a formalization of user privacy requirements in Section~\ref{section:policies}. 

\medskip
\noindent
\emph{C3) Information sensitivity.} In open settings, information sensitivity is is a function of user expectations and context: we therefore need to provide new definitions for sensitive information that takes user privacy requirements into account. We present context- and user-specific definitions of information sensitivity in Section~\ref{section:sensitive}.

\medskip
\noindent
\emph{C4) Adversarial knowledge estimation.} To adequately reason about disclosure risks in open settings we also require a parameterized adversary model that we can instantiate with various assumptions on the adversary's knowledge: this knowledge should include the information
disseminated by the user, as well as background knowledge to infer additional information about the user. In Section~\ref{section:general_model}, we define our adversary model
based on statistical inference.


In the following sections, we provide a rigorous formalization for these requirements, leading to a formal framework for privacy in open settings. We will instantiate this framework in Section~\ref{section:matching} to reason about the identity disclosure in particular.

We begin by discussing why existing privacy notions are not suited for reasoning about privacy in open settings. Afterwards, we provide an overview over further related work.
\subsection{Inadequacy of Existing Models}\label{section:general_idea}
Common existing privacy notions 
such as $k$-anonymity~\cite{Sweeney02}, $l$-diversity~\cite{Machanavajjhala07}, $t$-closeness~\cite{Li07} and the currently most 
popular notion of Differential Privacy~\cite{Dwork08} 
provide the technical means for privacy-friendly data-publishing in a closed-world setting: 
They target scenarios in which all data is available from the 
beginning, from a single data source, remains static and 
is globally sanitized in order to provide rigorous privacy guarantees. In what follows,
we describe how these notions fail to adequately address the challenges of
privacy in open settings discussed above.

\medskip
\noindent
\emph{a) Absence of structure and classification of data.} All the aforementioned privacy models
require an a-priori structure and classification of the data under consideration. Any
information gathered about an individual thus has to be embedded in this structure, or it cannot be seamlessly integrated in these models.

\medskip
\noindent
\emph{b) No differentiation of attributes.}
All of these models except for Differential Privacy require an additional differentiation between key attributes that identify an individual 
record, and sensitive attributes that a users seeks to protect. This again contradicts the absence of an a-priori, static 
structure in our setting. Moreover, as pointed out above and in the literature~\cite{Narayanan09}, such a differentiation cannot be 
made a-priori in general, and it would be highly context-sensitive in the open web setting.


\medskip
\noindent
\emph{c) Ubiquitously available background knowledge.}
All of these models, except for Differential Privacy, do not take into account adversaries that utilize
ubiquitously available background knowledge about a target user to infer additional sensitive information. A common example 
of background knowledge is openly available statistical information that allows the adversary to infer additional information 
about an identity. 

\medskip
\noindent
\emph{d) Privacy for individual users.}
All these models provide privacy for the whole dataset, which clearly implies privacy of every single user. 
One of the major challenges in open settings such as the Internet, however, is that accessing
and sanitizing all available data is impossible.
This leads to the requirement to design a local privacy notion that provides a lower privacy bound for 
every individual user, even if we only have partial access to the available data. 

\medskip
\noindent
The notion of Differential Privacy only fails to address some of the aforementioned requirements 
(parts \emph{a} and \emph{d}), but it comes with the additional assumption that
the adversary knows almost everything about the data set in question (everything
except for the information in one database entry). This assumption enables Differential Privacy to avoid differentiation between key attributes and sensitive attributes. This strong adversarial model, however, implies that privacy guarantees are only achievable if the considered data is globally perturbed~\cite{Dinur03,Dwork06,Dwork08-2}, which is not possible in open web settings. 

The conceptual reason for the inadequacy of existing models for reasoning about privacy in open web settings is mostly their design goal:
Privacy models have thus far mainly been concerned with the problem of attribute disclosure within a single data source: protection against identity disclosure was then
attempted by preventing the disclosure of any (sensitive) attributes of a user to the public. In contrast to static settings such as private data publishing, where we can decide which information will be disclosed to the adversary, protection against any attribute disclosure in open settings creates a very different set of challenges which we will address in the following sections.

\section{A Framework for Privacy in Open Settings}\label{section:general_model}
In this section, we first develop a user model 
that is suited for dealing with the information dissemination behavior commonly observed on 
the Internet. We then formalize our adversary model and show, as a sanity check, that hard privacy guarantees cannot 
be achieved in open settings. We conclude by defining privacy goals in open settings through user-specified privacy requirements 
from which we then derive a new definition of information sensitivity suited to open settings.

\subsection{Modeling Information in Open Settings}\label{section:data_model}

We first define the notion of entity models and restricted entity models. These models capture the behavior of these 
entities and in particular describe which attributes an entity exhibits publicly.
\begin{definition}[Entity Model] \label{def:statmodel}
Let $\Attributes$ be the set of all attributes. The \emph{entity model} $\theta_\entity$ 
of an entity $\entity$ provides for all attributes $\attribute\in\Attributes$ an \emph{attribute value} $\theta_\entity(\attribute)\in\dom(\attribute)\cup \{\mathsf{NULL}\}$ where $\dom(\attribute)$ is the domain over which the attribute $\attribute_i$ is defined.

The domain $\dom(\theta)$ of an entity model $\theta$ is the set of 
all attributes $\attribute\in\Attributes$ with value $\theta(\attribute) \neq \mathsf{NULL}$.  
\end{definition}
An entity model thus corresponds to the information an entity can publicly disseminate.
With the specific null value $\mathsf{NULL}$ we can also capture those cases where the entity
does not have any value for that specific attribute. 

In case the adversary has access to the full entity model, a set of entity models basically corresponds to a 
database with each attribute $\attribute\in\Attributes$ as its columns. In the open setting, however, an entity typically does 
not disseminate all attribute values, but instead only a small part of them. We capture this with the notion 
of restricted entity models.

\begin{definition}[Restricted Entity Model]
The \emph{restricted entity model} $\theta_\entity^{\Attributes'}$ is the entity model of \entity restricted 
to the non empty attribute set $\Attributes' \neq \emptyset$, i.e.,
\begin{equation*}
	\theta_\entity^{\Attributes'}(\attribute) = 
		\begin{cases}
			\theta_\entity(\attribute), & \text{if $\attribute\in\Attributes'$}\\
			\mathsf{NULL}, & \text{otherwise}
		\end{cases}
\end{equation*}
\end{definition}
In the online setting, each of the entities above corresponds to an online profile. A user $u$ usually uses more than one online service, each with different profiles $P_1^u,\ldots, P_l^u$. We thus define a user model as the collection of the entity models describing each of these profiles.

\begin{definition}[User Model / Profile Model]
The \emph{user model}\linebreak $\theta_u = \{\theta_{P_1^u,},\ldots,\theta_{P_1^u}\}$ of a user $u$ is 
a set of the entity models $\theta_{P_1^u,},\ldots,\theta_{P_1^u,}$, which we also call \emph{profile models}.
\end{definition}
With a user model that separates the information disseminated under different profiles, we will be able to formulate privacy requirements for each of these profiles separately. We will investigate this in Section~\ref{section:policies}.
\subsection{Adversary Model}\label{section:adversary_model}
In the following we formalize the adversary we consider for privacy in open settings. In our formalization, we follow the definitions of a semantic, Bayesian adversary introduced by Kasiviswanathan and Smith~\cite{Semantic_of_Differential_Privacy}.

For any profile $P$, we are interested in what the adversary $\adv$ learns about $P$ observing publicly available information from $P$. We formalize this learning process
through \emph{beliefs} on the models of each profile.

\begin{definition}[Belief]
Let $\Profiles$ be the set of all profiles and let $\mathcal{D}_\Attributes$ be the set of all distributions over profile models. 
A \emph{belief} $b = \{b_P | P\in\Profiles\}$ is a set of distributions $b_P\in\mathcal{D}_\Attributes$.  
\end{definition}
We can now define our privacy adversary in open settings using the notion of belief above.
\begin{definition}[Adversary]
An adversary $\adv$ is a pair of \emph{prior} belief $b$ and world knowledge $\Know$, i.e., $\adv = (b,\Know).$
\end{definition}
The adversary \adv's prior belief $b$ represents his belief in each profile's profile model before makes any observations. This prior belief can, in particular, also include background knowledge about each profile $P$. The world knowledge $\Know$ of the adversary represents a set of inference rules that allow him to infer additional attribute values about each profile from his observations.

We next define the publicly observations based on which the adversary learns additional information about each profile.
\begin{definition}[Publication Function]
A \emph{publication function} $G$ is a randomized function that maps each profile model $\theta_P$ to a restricted profile model 
$G(\theta_P) = \theta_P^{\Attributes'}$ such that there is at least one attribute $\attribute\in\Attributes'$ with 
$\theta_P(\attribute) = G(\theta_P)(\attribute)$.
\end{definition}
The publication function $G$ reflects which attributes are disseminated publicly by the user through his profile $P$. $G$ can, in particular, also include local sanitization where some attribute values are perturbed. However, we do require that at least one attribute value remains correct to capture utility requirements faced in open settings.

A public observation now is the collection of all restricted profile models generated by a publication function. 
\begin{definition}[Public Observation]
Let $\Profiles$ be the set of all profiles, and let $G$ be a publication function. 
The \emph{public observation} $\Obs$ is the set of all restricted profile models generated by $G$, i.e.,    
$\Obs = \{G(\theta_P)| P\in\Profiles\}$.  
\end{definition}
The public observation $\Obs$ essentially captures all publicly disseminated attribute values that can be observed by the adversary. 
Given such an observation $\Obs$, we can now determine what the adversary $\adv$ learns about each profile by determining his \emph{a-posteriori belief}. 

\begin{definition}[A-Posteriori Belief]
Let \Profiles be the set of all profiles. Given an adversary $\adv = (b,\Know)$ and a public observation $\Obs$, the adversary's \emph{a-posteriori belief} $\overline{b}=\{\overline{b}_P \in \mathcal{D}_\Attributes|P\in \Profiles\}$ is determined by applying the Bayesian inference rule, i.e.,
\[\overline{b}_P[\theta|\Obs, \Know] = \frac{Pr[ \Obs | \Know, \theta] \cdot b_P[\theta]}{\sum_{\theta'} Pr[ \Obs |\Know, \theta'] \cdot b_P[\theta']}.\]
\end{definition}
Here, the conditional probability $Pr[\Obs | \Know, \theta]$ describes the likelihood that the observational $\Obs$ is created by the specific entity model $\theta$. 

We will utilize the a-posteriori belief of the adversary to reason about the violation of the user specified privacy requirements in Section~\ref{section:policies}.

\subsection{Inapplicability of Statistical Privacy Notions}\label{section:inapplicability}
In the following, we formally show that traditional non-disclosure guarantees, e.g., in the style of Differential Privacy, are not possible in open settings. 

Kasiviswanathan and Smith~\cite{Semantic_of_Differential_Privacy} provide a general definition of 
non-disclosure they call $\epsilon$-privacy. In their definition, they compare the adversary $\adv$'s a-posteriori beliefs 
after observing the transcript $t$ generated from a database sanitazitaion mechanism $\mathcal{F}$ applied 
on two adjacent databases with n rows: first on the database $x$, leading to the belief $\overline{b}_0[.|t]$, 
and secondly on the database $x_{-i}$, where a value in the $i$th row in $x$ is replaced by a default value, 
leading to the belief $\overline{b}_i[.|t]$.
\begin{definition}[$\epsilon$-semantic Privacy~\cite{Semantic_of_Differential_Privacy}]
Let $\epsilon\in[0,1]$. A randomized algorithm $\mathcal{F}$ is $\epsilon$-semantically private if for all belief distributions $b$ on $D^n$, for all possible transcripts, and for all $i=1\ldots n$: \[\mathsf{SD}(\overline{b}_0[.|t], \overline{b}_i[.|t]) \leq \epsilon.\]
\end{definition}
Here, $\mathsf{SD}$ is the total variation distance of two probability distributions. 

\begin{definition}
Let $X$ and $Y$ be two probability distributions over the sample space $D$. 
The \emph{total variation distance} $\mathsf{SD}$ of $X$ and $Y$ is
\[\mathsf{SD}(X,Y) = \mathsf{max}_{S\subset D}\left[ Pr[X\in S] - Pr[Y\in S] \right].\]
\end{definition}
Kasiviswanathan and Smith~\cite{Semantic_of_Differential_Privacy} show that $\epsilon$-differential privacy is essentially equivalent to $\epsilon$-semantic privacy. 

In our formalization of privacy in open settings, varying a single database entry corresponds to changing the value of a single 
attribute $\attribute$ in the profile model $\theta_P$ of a profile $P$ to a default value. We denote this 
modified entity model with $\theta_P^\attribute$,
and the thereby produced a-posteriori belief by $\overline{b}_P^\attribute$. 
A profile $P$ would then be $\epsilon$-semantically private if for any modified profile 
model $\theta_P^\attribute$, the a-posteriori belief of adversary $\adv$ does not change by more than $\epsilon$. 

\begin{definition}[$\epsilon$-semantic Privacy in Open Settings]
Let $\epsilon\in[0,1]$. A profile $P$ is $\epsilon$-semantically private in open settings if for any attribute $\attribute$,
\[\mathsf{SD}(\overline{b}_P[.|\Obs], \overline{b}_P^\attribute[.|\Obs]) \leq \epsilon\]
where $\overline{b}_P$ and $\overline{b}_P^\attribute$ are the a-posteriori beliefs of the adversary after 
observing the public output of $\theta_P$ and $\theta_P^\attribute$ respectively.
\end{definition}
As expected, we can show that $\epsilon$-semantic privacy can only hold for $\epsilon = 1$ in open settings.
\begin{theorem}
For any profile model $\theta_P$ and any attribute $\attribute$, there is an adversary $\adv$ 
such that \[\mathsf{SD}(\overline{b}[.|\Obs], \overline{b}^\attribute[.|\Obs]) \geq 1.\]
\end{theorem}
\begin{proof}
Let $\adv$ have a uniform prior belief, i.e., all possible profile models have the same probability, and empty world knowledge $\Know$. 
Let \attribute be the one attribute that remains the same after applying the publication function $G$.
Let $x$ be the original value of this attribute $\attribute$ and let $x^\ast$ be the default value that replaces $x$. 

Observing the restricted profile model $\theta_P[\Attributes']$ without any
additional world knowledge will lead to an a-posteriori belief, 
where the probability of the entity model $\theta$ with $\theta[\Attributes'] = \theta_P[\Attributes']$ and 
$\mathsf{NULL}$ everywhere else, is set to 1. 

Conversely, the modified setting will result in an a-posteriori belief that sets the probability for the entity model 
$\theta^\ast$ to one, where $\theta^\ast$ is constructed for the modified setting as $\theta$ above. 
Thus $\overline{b}[\theta|\Obs] = 1$, whereas $\overline{b}^\attribute[\theta|\Obs] = 0$, and hence $\mathsf{SD}(\overline{b}[.|\Obs], \overline{b}^\attribute[.|\Obs]) = 1$. 
\qed
\end{proof}

Intuitively, the adversary can easily distinguish differing profile models because a) he can directly observe the profiles publicly available information, b) he chooses which attributes he considers for his inference and c) only restricted, local sanitization is available to the profile. Since these are elementary properties of privacy in open settings, we can conclude that hard security guarantees in the style of differential privacy are impossible to achieve in open settings. 

However, we can provide an assessment of the disclosure risks by explicitly fixing the a-priori knowledge and the attribute set considered by the adversary. While we no longer all-quantify over all possible adversaries, and therefore lose the full generality of traditional non-disclosure guarantees, we might still provide meaningful privacy assessments in practice. We further discuss this approach in Section~\ref{section:sensitive}, and follow this approach in our instantiation of the general model for assessing the likelihood of identity disclosure in Section~\ref{section:linkability_model}.

\subsection{User-Specified Privacy Requirements}\label{section:policies}
In the following we introduce user-specified privacy requirements that allow us to formulate privacy goals that 
are user- and context-dependent. These can then lead to restricted privacy assessments instead of general privacy guarantees that we have shown to be impossible in open setting in the previous section.  

%
We define a user's privacy requirements on a per-profile basis, stating which attribute values should not be inferred by adversary after seeing a public observations \Obs. 

\begin{definition}[Privacy Policy]
A \emph{privacy policy} $\Policy$ is a set of privacy requirements $\policy = (P, \{\attribute_i = x_i\})$ which require that profile $P$ should never expose the attribute values $x_i$ for the attributes $\attribute_i \in \Attributes$.
\end{definition}
By setting privacy requirements in a per-profile basis we capture an important property of information dissemination in open settings: users utilize different profiles for different context (e.g., different online services) assuming these profiles remain separate and specific information is only disseminated under specific circumstances.

Given the definition of privacy policies, we now define the violation of a policy by considering the adversary's a-posterior belief $\overline{b}$, as introduced in Section~\ref{section:adversary_model}. 

\begin{definition}[Privacy Policy Satisfaction / Violation]
Let $\adv = (b,\Know)$ be an adversary with a-posteriori belief $\overline{b}$, and let $\theta[\attribute=x]$ be the set of all entity models that have the value $x$ for the attribute $\attribute$.
A profile $P_i^u$ \emph{$\sigma$-satisfies} a user's privacy requirement $\policy_j^u = (P, \{\attribute_i = x_i\})$, written $P_i^u \models_\sigma \policy_j^u$, if
\begin{itemize}
\item $P = P_i^u$
\item $\forall \attribute_i: \sum_{\theta\in\theta[\attribute_i=x_i]} \overline{b}_P[\theta|\Obs,\Know] \leq \sigma$
\end{itemize} 
and \emph{$\sigma$-violates} the user's privacy requirement otherwise.

A user model $\theta_{u}$ \emph{$\sigma$-satisfies} a user $u$'s privacy policy $\Policy_u$, written $\theta_{u} \models_\sigma \Policy_u$, if all profile models $\theta_{P_i^u}$ $\sigma$-satisfy their corresponding privacy requirements, and \emph{$\sigma$-violates} the privacy policy otherwise.
\end{definition} 
The above attributes can also take the form of ``$P$ belongs to the same user as $P'$'', effectively restricting which profiles should be linked to each other. We will investigate this profile linkability problem specifically in Section~\ref{section:linkability_model}.

\subsection{Sensitive Information}\label{section:sensitive}
In contrast to the closed-world setting, with its predefined set of sensitive attributes that 
automatically defines the privacy requirements, a suitable definition of information 
sensitivity in the open setting is still missing. 
In the following, we derive the notion of sensitive information from the user privacy requirements we defined in Section~\ref{section:policies}. 


\begin{definition}[Sensitive Attributes]
A set of attributes $\Attributes^\ast$ is \emph{sensitive} for a user $u$ in the context of her profile $P_i^u$ if $u$'s privacy policy $\Policy_u$ 
contains a privacy requirement $\policy = (P_i^u,\Attributes' = X)$ where $\Attributes^\ast\subseteq\Attributes'$.
\end{definition}
Here, we use the notation $\Attributes = X$ as vector representation for $\forall \attribute_i \in \Attributes: \attribute_i = x_i$.

Sensitive attributes, as defined above, are not the only type of attributes that are worth to protect: In practice, an adversary can additionally infer sensitive attributes from other attributes through statistical inference using a-priori knowledge. We call such attributes that allow for the inference of sensitive attributes \emph{critical attributes}. 

\begin{definition}[Critical Attributes]
Given a set of attributes $\Attributes^\ast$, let $P$ be a profile with $\dom(\theta_P) \supseteq \Attributes$, and let $P'$ be the profile with the restricted profile model $\theta_{P'} = \theta_P^{\Attributes'}$, where 
$\Attributes' = \dom(\theta_P) \setminus \Attributes^\ast$.

The set of attributes $\Attributes^\ast$ is \emph{$\sigma$-critical} for the user $u$ that owns the profile $P$ and an adversary with prior belief $b_P$ and world knowledge $\Know$, if $u$'s privacy policy $\Policy_u$ contains a privacy requirement 
$\policy$ such that $P$ $\sigma$-violates $\policy$ but $P'$ does not. 
\end{definition}
Critical information require the same amount of protection as sensitive information, the difference however being that critical information is only protected for the sake of protecting sensitive information. 

As a direct consequence of the definition above, sensitive attributes are also critical.
\begin{corollary}\label{cor:sens_crit}
Let \Attributes be a set of sensitive attributes. Then \Attributes is also $0$-critical.
\end{corollary}
Another consequence we can draw is that privacy requirements will always be satisfied if no critical attributes are disseminated.
\begin{corollary}\label{cor:assessment}
Let \Obs be a public observations that does not include any critical attributes for a user $u$ and an adversary $\adv$. Then $u$'s privacy policy $\Policy_u$ is $\sigma$-satisfied against $\adv$.
\end{corollary}
The corollary above implies that, while we cannot provide general non-disclosure guarantees in open settings, we can provide privacy assessments for 
specific privacy requirements, given an accurate estimate of the adversary's prior beliefs.\\\\ 
While privacy assessments alone are not satisfactory from a computer 
security perspective, where we usually require hard security guarantees quantified over all possible adversaries, the fact remains that we are faced with privacy 
issues in open settings that are to this day unanswered for due to the impossibility of hard guarantees in such settings. Pragmatically thinking, we are convinced
that we should move from impossible hard guarantees to more practical privacy assessments instead. 
This makes particularly sense in settings where users are not victims of targeted attacks, but instead fear
attribute disclosure to data-collecting third parties.   


\section{Linkability in Open Settings}\label{section:linkability_model}
In the following we instantiate the general privacy model introduced in the last section to reason about the likelihood that two profiles of the same user are linked by the adversary in open settings. We introduce the novel notion of $(k,d)$-anonymity with which we assess anonymity and linkability based on the similarity of profiles within an online community.

To simplify the notation we introduce in this section, we will, in the following, talk about matching \emph{entities} $\entity$ and $\entity'$ the adversary wants to link, instead of profiles $P_1$ and $P_2$ that belong to the same user $u$. All definitions introduced in the general framework above naturally carry over to entities as well.

\subsection{Model Instantiation for Linkability}
In the linkability problem, we are interested in assessing the likelihood that two matching entities $\entity$ and $\entity'$ can be linked, potentially across different online platforms. The corresponding privacy requirements, as introduced in Section~\ref{section:policies}, are $r_1 = (\entity,\attribute_L)$ and $r_2 = (\entity', \attribute_L)$, where $\attribute_L$ is the attribute that $\entity$ and $\entity'$ belong to the same user. Consequently, we say that these entities are unlinkable if they satisfy the aforementioned privacy requirements.

\begin{definition}[Unlinkability]
Two entities $\entity$ and $\entity'$ are $\sigma$-unlinkable if \\$\{\theta_\entity,\theta_{\entity'}\} \models_\sigma \{r_1, r_2\}$.
\end{definition}

\subsection{Anonymity}\label{sec:anonymity}
To assess the identity disclosure risk of an entity \entity within a collection of entities \Entities, we use the following intuition: $\entity$ is anonymous in \Entities 
if there is a subset $\Entities' \subseteq \Entities$ to which \entity is very similar. The collection 
\Entities' then is an anonymous subset of \Entities for \entity.

To assess the similarity of entities within a collection of entities, we will use a distance measure $\dist$ on the entity models of these entities. We will require that this measure provides all properties of a metric.

A collection of entities in which the distance of all entities to \entity is small (i.e., $\leq$ a constant $d$) is called $d$-convergent for \entity. 

\begin{definition}
A \emph{collection} of entities $\Entities$ is $d$-convergent for \entity
if $\dist(\theta_{\entity}, \theta_{\entity'}) \leq d$ for all $\entity'\in \Entities$.
\end{definition}
Convergence measures the similarity of a collection of individuals. Anonymity is achieved if an entity can find a collection of entities that are all similar to this entity. This leads us to the definition of $(k,d)$-anonymity, which requires a subset of similar entities of size $k$.

\begin{definition}\label{def:kd-anonymity}
An entity $\entity$ is \emph{$(k,d)$-anonymous} in a collection of entities $\Entities$ if there exists a subset of entities $\mathcal \Entities' \subseteq \Entities$ with the properties 
that $\entity \in \Entities$, that $\left| \Entities' \right| \geq k$ and that $\Entities'$ is $d$-convergent.
\end{definition}
An important feature of this anonymity definition is that it provides anonymity guarantees that can be derived from a subset of all available data, but continue to hold once we consider a larger part of the dataset.
\begin{corollary}
If an entity is $(k,d)$-anonymous in a collection of entities \Entities, then it is also $(k,d)$-anonymous in the collection  of entities $\Entities' \supset \Entities$.
\end{corollary} 
Intuitively, $(k,d)$-anonymity is a generalization of the classical notions of $k$-anonymity to open settings without pre-defined quasi-identifiers. We schematically illustrate such anonymous subsets in Figure~\ref{fig:anonymity_settings}.
\begin{figure}[!t]
	\centering
    \includegraphics[width=0.7\columnwidth]{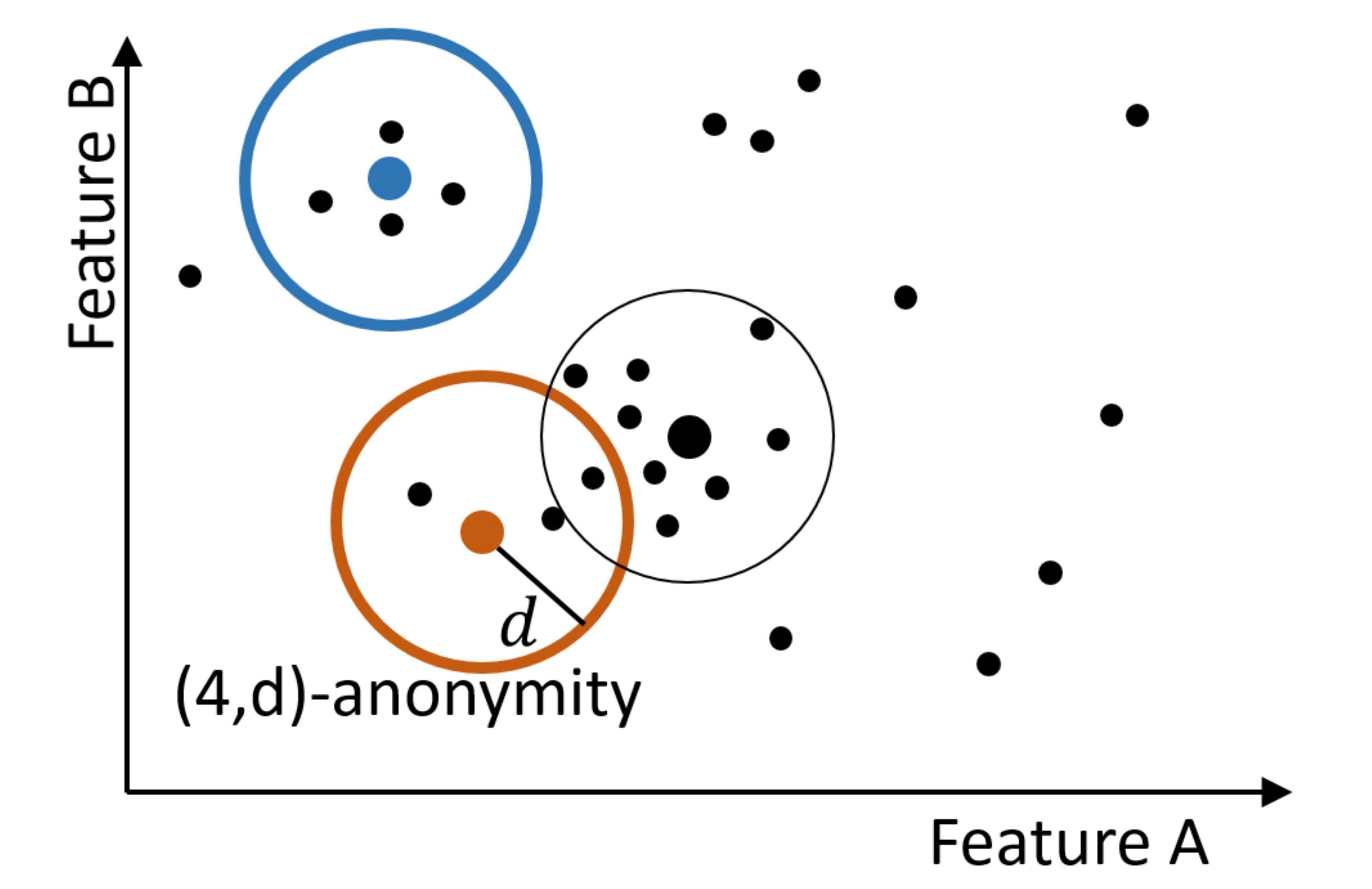}
\caption{Anonymity in crowdsourcing systems.}\label{fig:anonymity_settings}
\vspace*{-2ex}
\end{figure}


\subsection{Entity Matching}\label{section:matching}
We define the notion of \emph{matching} identities. As before, we use the distance measure $\dist$ to assess the similarity of two entities.
\begin{definition}
An entity $\entity$ $c$-matches an entity $\entity'$ if $\dist(\theta_{\entity}, \theta_{\entity'}) \leq c$.
\end{definition}
Similarly, we can also define the notion of one entity matching a collection of entities.
\begin{definition}
A collection of entities $\Entities$ \emph{$c$-matches} an entity $\entity'$ if all entities $\entity \in \Entities$ 
$c$-match $\entity'$.
\end{definition}
Assuming the adversary only has access to the similarity of entities, the best he can do is comparing the distance of all 
entities $\entity\in\Entities$ to $\entity'$ and make a probabilistic choice proportional to their relative distance values.

Now, if the matching identity $\entity^\ast$ is $d$-convergent in $\Entities$ the, all entities in $\Entities$ will have a comparatively similar distance to $\entity'$. 
\begin{lemma}\label{lem:convergent_matching}
Let $\Entities$ be $d$-convergent for $\entity^\ast$.
If $\entity^\ast$ $c$-matches $\entity'$, then $\Entities$ $(c+d)$-matches $\entity'$.
\end{lemma}
\begin{proof}
Since $\Entities$ is $d$-convergent for $\entity^\ast$, $\forall \entity'\in\Entities:\ \dist(\entity^\ast, \entity') \leq d$.
Using the triangle inequality, and the fact that $\entity^\ast$ $c$-matches the entity $\entity'$, we can bound the distance of all entities $\entity\in\Entities$ to $\entity'$ by $\forall \entity''\in\Entities:\ \dist(\entity, \entity') \leq c+d$. 
Hence $\Entities$ $(c+d)$-matches the entity $\entity'$. \qed
\end{proof}
Hence, the matching entity $\entity^\ast$ does not $c$-match $\entity'$ for a small value of $c$, 
the adversary $\adv$ he will have a number of possibly matching entities that are similarly likely to match $\entity'$.

We get the same result if not the whole collection $\Entities$ is convergent, but if there exists a subset of convergent entities that allows the target to remain anonymous.
\begin{corollary}
Let \entity' be $(k,d)$-anonymous in $\Entities$. If \entity' $c$-matches an entity \entity then there is a subset $\Entities' \subseteq \Entities$ of size at least $k$ which $(c+d)$-matches \entity. 
\end{corollary}

%
%

\subsection{Identity Disclosure}\label{section:identity}
We assume that the adversary uses the similarity of the candidate entities to his target entity $\entity'$ to make his decision. The likelihood that the adversary chooses a specific entity $\entity^\ast$ then is 
the relative magnitude of $\dist(\entity^\ast,\entity)$, i.e. 
\[Pr[\Adv\ \text{chooses}\ \entity^\ast] = 1 - \frac{\dist(\entity^\ast, \entity')}{\sum_{\entity\in\Entities}\dist(\entity,\entity')}.\]
We can now bound the likelihood with which a specific entity $\entity^\ast$ would be chosen by the adversary if 
$\entity^\ast$ is $(k,d)$-anonymous.
\begin{theorem}\label{thm:identity}
Let the matching entity $\entity^\ast$ of the entity $\entity'$ in the collection $\Entities = \{\entity_1,\ldots, \entity_n\}$ be $(k,d)$-anonymous in $\Entities$. 
Furthermore let $\entity^\ast$ $c$-match  $\entity'$. Then an adversary $\adv = (b,\emptyset)$ with uniform prior belief $b$ and with empty world knowledge that only observes the similarity of entities links the entity $\entity^\ast$ to $\entity'$ with a likelihood of at most $t \leq 1- \frac{c}{c+(k-1)(c+d)}$. 
\end{theorem}
\begin{proof}
Let $\Entities^{\ast}$ be the $(k,d)$ anonymous subset of $\entity^{\ast}$ in \Entities. Let $t^\ast$ be the likelihood of identifying $\entity^\ast$ from $\Entities^{\ast}$. 
Then clearly $t < t^\ast$ since we remove all possible, but wrong candidates in $\Entities \setminus \Entities^\ast$.

Since $\entity^\ast$ $c$-matches $\entity'$, by Lemma~\ref{lem:convergent_matching}, we can upper bound the distance of each entity in $\Entities^\ast$ to $\entity'$, 
i.e.,
\[\forall \entity \in \Entities^\ast: \dist(\entity, \entity') \leq c+d\]
\noindent
We can now bound $t^\ast$ as follows:
\begin{align*}
t^\ast &= Pr[\Adv\ \text{chooses}\ \entity] \\
 &= 1 - \frac{c}{c + (k-1)(\sum\limits_{\entity \in \Entities^\ast \setminus \{\entity^\ast \}} \dist(\entity, \entity') )} \leq 1 - \frac{c}{c + (k-1)(c+d)}
\end{align*}
\qed
\end{proof}
Theorem~\ref{thm:identity} shows that, as long as entities remain anonymous in a suitably large anonymous subset of a collection of entities, an adversary will have difficulty identifying them 
with high likelihood. Recalling our unlinkability definition from the beginning of the section, this result also implies that $\entity^\ast$ is $\sigma$-unlinkable for $\sigma = t$.
\begin{corollary}
Let the matching entity $\entity^\ast$ of the entity $\entity'$ in the collection $\Entities = \{\entity_1,\ldots, \entity_n\}$ be $(k,d)$-anonymous in \Entities. Then $\entity^\ast$ and $\entity'$ are $\sigma$-unlinkable for $\sigma = 1- \frac{c}{c+(k-1)(c+d)}$ against an adversary $\adv=(b,\emptyset)$ with uniform prior belief and empty world knowledge that only observes entity similarity.
\end{corollary}
In Section~\ref{section:data_evaluation} we present experiments that evaluate the anonymity and linkability of individuals in the Online Social Network Reddit, and measure how well they can be identified from among their peers.
\subsection{Limitations}\label{section:limitations}
The quality of the assessment provided by the $d$-convergence model largely depends on adversarial prior belief we assume: in our results above, we assume an adversary without any prior knowledge. 
In practice, however, the adversary might have access to prior beliefs that can help him in his decision making. Therefore, turning such assessments into meaningful estimates in practice requires a careful estimation of prior knowledge by, e.g., producing a more accurate profile model: the problem of comprehensive profile building for entities in an open setting is an open question that has been examined somewhat in the literature~\cite{Cali04,Cortis12,Chen12,Sharma12,Balduzzi10}, but on the whole still leaves a lot of space for future work.   
\\\\
This concludes the formal definitions of our $d$-convergence model. In the next sections, we instantiate it for identity disclosure risk analyses based on user-generated text-content and apply this instantiation to the OSN Reddit. 


\section{Linkability Evaluation on Reddit}\label{section:experiment}
While the main focus of this paper is to present the actual privacy model as such, the following experiments are meant to provide first insights into the application of our framework, without taking overly complex adversarial capabilities
into account. The evaluation can easily be extended to a more refined model of an adversary
without conceptual difficulties.

We first articulate the goals of this evaluation, and then, secondly, describe the data collection process, followed by defining the instantiation 
of the general framework we use for our evaluation in the third step.
Fourth, we introduce the necessary processing steps on our dataset, before we finally discuss the results of our evaluation.


\subsection{Goals}\label{section:hypotheses}
In our evaluation, we aim at validating our model by conducting two basic experiments.
First, we want to empirically show that, our model instantiation yields a suitable abstraction of real users for reasoning about their privacy.
To this end, profiles of the same user should be more similar to each other (less distant) than profiles from different users.

Second, we want to empirically show that a larger anonymous subset makes it more difficult for an adversary to correctly link the profile. Thereby, we inspect whether anonymous subsets provide a practical estimate of a profile's anonymity.

Given profiles with anonymous subsets of similar size, we determine the percentage of profiles which the adversary can match within the top $k$ results, i.e., given a source profile, the adversary computes the top $k$ most similar (less distant) profiles in the other subreddit. We denote this percentage by \emph{precision@$k$} and correlate it to the size of the anonymous subsets.

We fix the convergence of the anonymous subsets to be equal to the matching distance between two corresponding profiles. Our intuition is that, this way, the anonymous subset captures most of the profiles an adversary could potentially consider matching.
\subsection{Data-Collection}\label{section:data_collection}
For the empirical evaluation of our privacy model, we use the online social network Reddit~\cite{reddit} that was founded in 2005 and constitutes one of the largest discussion and information sharing platforms in use today.
On Reddit, users share and discuss topics in a vast array of topical subreddits that collect all topics belonging to one general area; e.g. there are subreddits for world news, tv series, sports, food, gaming and many others. Each subreddit contains so-called submissions, i.e., user-generated content that can be commented on by other users.

To have a ground truth for our evaluation, we require profiles of the same user same user across different OSNs to be linked. Fortunately, Reddit's structure provides an inherent mechanism to deal with this requirement.
Instead of considering Reddit as a single OSN, we treat each subreddit as its own OSN. Since users are identified through the same pseudonym in all of those subreddits, they remain linkable across
the subreddits' boundaries. Hence our analysis has the required ground truth. The adversary we simulate, however, is only provided with the information available in the context of each subreddit and thus can only try to match profiles across subreddits. Ground truth in the end allows us to verify the correctness of his match.


To build up our dataset, we built a crawler using Reddit's API to collect comments. Recall that subreddits contain submissions that, in turn, are commented by the users.
For our crawler, we focused on the large amount of comments because they contain a lot of text and thus are best suitable for computing the unigram models. 

Our crawler operates in two steps that are repeatedly executed over time. During the whole crawling process, it maintains a list of already processed users. In the first step, our crawler collects a list of the overall newest comments on Reddit from Reddit's API and inserts these comments into our dataset. In the second step, for each author of these comments who has not been processed yet, the crawler also collects and inserts her latest $1,000$ comments into our dataset. Then, it updates the list of processed users. The number of $1,000$ comments per user, is a restriction of Reddit's API.

In total, during the whole September 2014, we collected more than $40$ million comments from over $44,000$ subreddits. The comments were written by about $81,000$ different users which results in more than $2.75$ million different profiles.

The whole dataset is stored in an anonymized form in a MySQL database and is available upon request.

\subsection{Model Instantiation}\label{section:instantiation}
On Reddit, users only interact with each other by by posting comments to text of link submissions. Reddit therefore does not allow us to exploit features found in other social networks, such as friend links or other static data about each user. On the other hand, this provides us with the opportunity to evaluate the 
linkability model introduced in Section~\ref{section:linkability_model} based dynamic, user-generated content, in this case user-generated text content.

Since we only consider text content, we instantiate the general model from the previous sections with an unigram model, where each attribute is a word unigram, an its value is the frequency with which the unigram appears in the profiles comments. Such unigram models have succesfully been used in the past to characterize the information within text content and to correlate users across different online platforms~\cite{mywww,MishariT12}.
\begin{definition}[Unigram Model]
Let $\vocabulary$ be a finite vocabulary. The \emph{unigram model} $\theta_P = {p_i}$ of a profile is a set of \emph{frequencies} $p_i\in[0,\ldots,1]$ with which each unigram $w_i\in\vocabulary$ appears in the profile $P$. Each frequency $p_i$ is determined by
\[p_i = \frac{\stringCount(w_i,P)}{\sum_{w\in\vocabulary}{\stringCount(w,P)}}\]  
\end{definition}
Since the unigram model essentially constitutes a probability distribution, we instantiate our distance metric $\dist$ with the Jensen-Shannon divergence~\cite{Endres03}. The Jensen-Shannon divergence is a symmetric extension of the Kullback-Leiber divergence has been shown to be successful in many related information retrieval scenarios.
\begin{definition}\label{def:divergence}
Let $P$ and $Q$ be two statistical models over a discrete space $\Omega$. 
The Jensen-Shannon divergence is defined by 
	\[\div_{\mathsf{JS}} = \frac{1}{2}\div_{\mathsf{KL}}(P||M) + \frac{1}{2}\div_{\mathsf{KL}}(Q||M) \]
where $\div_{\mathsf{KL}}$ is the Kullback-Leibler divergence  
\[\div_{\mathsf{KL}}(P||Q) = \sum_{\omega\in\Omega} log\left(\frac{P(\omega)}{Q(\omega)}\right)P(\omega)\]
and $M$ is the averaged distribution $M = \frac{1}{2}(P+Q)$.
\end{definition}
In the following, we will use the square-root of the Jensen-Shannon divergence, constituting a metric, as our distance measure, i.e., $\dist=\sqrt{\div_{\mathsf{JS}}}$.


\subsection{Data-Processing}\label{section:data_processing}
The evaluation on our dataset is divided into sequentially performed computation steps, which include
the normalization of all comments, the computation of unigram models for each profile, a filtering of our dataset to keep the evaluation tractable, the computation of profile distances and the computation of $(k,d)$-anonymous subsets.

\paragraph{Normalizing Comments.}
Unstructured, heterogeneous data, as in our case, may contain a variety of valuable information about a user's behavior, e.g., including formatting and punctuation.
Although we could transform these into attributes, we do not consider them here for the sake of simplicity.

In order to get a clean representation to apply the unigram model on, we apply various normalization steps, including transformation to lower case, the removal of Reddit formatting and punctuation except for smilies.
Moreover, we apply a encoding specific normalization, replace URLs by their hostnames and shorten repeated characters in words like \reddit{cooool} to a maximum of three. 
Finally, we also filter out a list of $597$ stopwords from the comments.
Therefore, we perform six different preprocessing steps on the data, which we describe in more detail in the following.

\begin{enumerate}
\item \textbf{Convert to lower case letters:} In our statistical language models, we do not want to differentiate between capitalized and lowercased occurrences
    of words. Therefore, we convert the whole comment into lower case.

\item \textbf{Remove Reddit formatting:} Reddit allows users to use a wide range of formatting modifiers that we divide into two basic categories: 
    formatting modifiers that influence the typography and the layout of the comment, and formatting modifiers that include external resources
    into a comment. The first kind of modifier, named layout modifiers, is stripped off the comment, while leaving the plain text.
    The second kind of modifier, called embedding modifiers, is removed from the comment completely.

    One example for a layout modifier is the asterisk: When placing an asterisk both in front and behind some text, e.g., \reddit{*text*},
    this text will be displayed in italics, e.g., \reddit{\textit{text}}.
    Our implementation removes these enclosing asterisks, because they are not valuable for computing statistical language models for $n$-grams and only
    affect the layout. Similarly, we also remove other layout modifiers such as table layouts, list layouts and URL formatting in a way that only the important information remains.

    A simple example for embedding modifiers are inline code blocks: Users can embed arbitrary code snippets into their comments using the \reddit{`} modifier.
    Since these code blocks do not belong to the natural language part of the comment and only embed a kind of external resource, we remove them completely.
    In addition to code blocks, the category of embedding modifiers also includes quotes of other comments.

\item \textbf{Remove stacked diacritics:} In our dataset, we have seen that diacritics are often misused. Since Reddit uses Unicode as its character encoding,
    users can create their own characters by arbitrarily stacking diacritics on top of them. 
    To avoid this kind of unwanted characters, we first normalize the comment by utilizing the unicode character composition, which tries to combine each letter and
    its diacritics into a single precombined character. Secondly, we remove all remaining diacritic symbols from the comment. While this process preserves
    most of the normal use of diacritics, it is able to remove all unwanted diacritics.

\item \textbf{Replace URLs by their hostname:} Generally, a URL is very specific and a user often does not include the exact same URL in different comments.
    However, it is much more common that a user includes different URLs that all belong to the same hostname, e.g., \reddit{www.mypage.com}.
    Since our statistical language models should represent the expected behavior of a user in terms of used words (including URLs), we restrict all URLs to their hostnames.

\item \textbf{Remove punctuation:} Most of the punctuation belongs to the sentence structure and, thus, should not a part of our statistical language
    models. Therefore, we remove all punctuation except for the punctuation inside URLs and smilies. We do not remove the smilies, because people are using
    them in a similar role as words to enrich their sentences: Every person has her own subset of smilies that she typically uses. To keep the smilies in the comment,
    we maintain a list of $153$ different smilies that will not be removed from the comment.

\item \textbf{Remove duplicated characters:} In the internet, people often duplicate characters in a word to add emotional nuances to their writing, e.g.,
    \reddit{cooooo\-oooool}. But sometimes the number of reduplicated characters varies, even if the same emotion should be expressed. Thus, we reduce
    the number of duplicated characters to a maximum of $3$, e.g., \reddit{coool}. In practice, this truncation allows us to differentiate between
    the standard use of a word and the emotional variation of it, while it does not depend on the actual number of duplicated characters.
\end{enumerate}

\begin{table*}[t]
  \begin{center}
    \begin{tabular}{| l || c | c || c | c || c | c |}
    \hline
    Top & \multicolumn{2}{c||}{Reddit} & \multicolumn{2}{c||}{subreddit: Lost} & \multicolumn{2}{c|}{subreddit: TipOfMyTongue} \\
    \hline
    & Unigram & Frequency & Unigram & Frequency & Unigram & Frequency \\
    \hline
    \hline
    1. & people & 4,127,820 & island & 832 & www.youtube.com & 3663 \\
    2. & time & 2,814,841 & show & 750 & song & 1,542 \\
    3. & good & 2,710,665 & lost & 653 & remember & 1,261 \\
    4. & gt & 2,444,240 & time & 580 & en.wikipedia.org & 1,100 \\
    5. & game & 1,958,850 & people & 527 & sounds & 1,007 \\
    6. & pretty & 1,422,640 & locke & 494 & solved & 924 \\
    7. & 2 & 1,413,118 & season & 431 & movie & 918 \\
    8. & lot & 1,385,167 & jacob & 429 & find & 829 \\
    9. & work & 1,352,292 & mib & 372 & :) & 786 \\
    10. & 1 & 1,184,029 & jack & 310 & game & 725 \\
    11. & 3 & 1,124,503 & episode & 280 & time & 678 \\
    12. & great & 1,070,299 & ben & 255 & thinking & 633 \\
    13. & point & 1,063,239 & good & 250 & good & 633 \\
    14. & play & 1,060,985 & monster & 237 & www.imdb.com & 584 \\
    15. & years & 1,032,270 & lot & 220 & video & 583 \\
    16. & bad & 1,008,607 & gt & 182 & pretty & 570 \\
    17. & day & 989,180 & character & 165 & youtu.be & 569 \\
    18. & love & 988,567 & walt & 163 & mark & 548 \\
    19. & find & 987,171 & man & 162 & edit & 540 \\
    20. & shit & 976,928 & dharma & 162 & post & 519 \\
    \hline
    \end{tabular}
  \end{center}
  \caption{Top 20 unigrams of Reddit and two sample subreddits Lost and TipOfMyTongue.}
  \label{tbl:unigrams}
\end{table*}
\paragraph{Computing Unigram Models.}
\label{sec:unigram_frequencies}
From the normalized data, we compute the unigram frequencies for each comment. Recall that our dataset consists of 
many subreddits that each form their own OSN. Thus, we aggregate the corresponding unigram frequencies per profile,
per subreddit, and for Reddit as a whole. 
Using this data, we compute the word unigram frequencies for each comment as described in Section~\ref{section:instantiation}.

Since a subreddit collects submissions and comments to a single topic, we expect the unigrams to reflect its topic specific language. Indeed, the 20 most frequently used unigrams of a subreddit demonstrate that the language adapts to the topic.
As an example, we show the top 20 unigrams (excluding stopwords) of Reddit and two sample subreddits \emph{Lost} and \emph{TipOfMyTongue} in Table \ref{tbl:unigrams}.
As expected, there are subreddit specific unigrams that occur more often in the context of one subreddit than in the context of any other subreddit. For example, the subreddit \emph{Lost} deals with a TV series that is about the survivors of a plane crash and its aftermath on an island. Unsurprisingly, the word \emph{island} is the top unigram in this subreddit. In contrast, the subreddit \emph{TipOfMyTongue} deals with the failure to remember a word from memory and, thus, has the word \emph{remember} in the list of its top three unigrams. 


\paragraph{Filtering the Dataset.}
To reduce the required amount of computations we restrict ourselves to \emph{interesting profiles}. We define an interesting profile as one that contains at least $100$ comments and that belongs to a subreddit with at least $100$ profiles.
Additionally, we dropped the three largest subreddits from our dataset to speed up the computation.

In conclusion, this filtering results in $58,091$ different profiles that belong to $37,935$ different users in $1,930$ different subreddits.

\paragraph{Distances Within and Across Subreddits.}
Next, we compute the pairwise distance within and across subreddits using our model instantiation.
Excluding the distance of profiles to themselves, the minimal, maximal and average distance of two profiles within subreddits in our dataset are approximately $0.12$, $1$ and $0.79$ respectively.
Across subreddits, the minimal, maximal and average distance of two profiles are approximately $0.1$, $1$ and $0.85$ respectively.

\paragraph{Anonymous Subsets.}
Utilizing the distances within subreddits, we can determine the anonymous subsets for each profile in a subreddit. 
More precisely, we compute the anonymous subset for each pair of profiles from the same user. We set the convergence $d$ to the matching distance between both profiles and determine the size of the resulting anonymous subset.
\begin{figure*}[!t]
    \centering
        \includegraphics[width=0.8\textwidth]{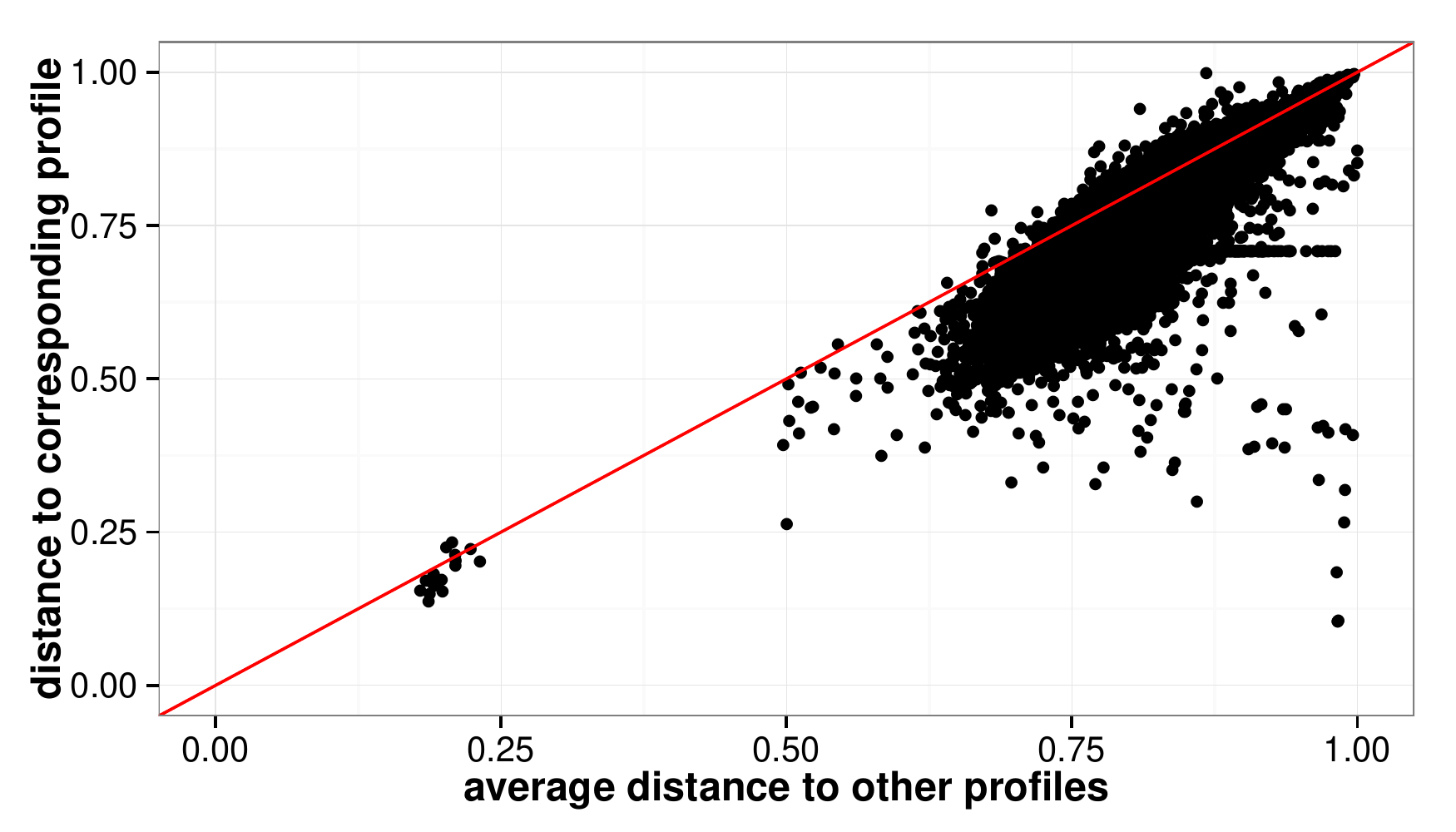}
        \caption{The average distance between a profile in subreddit $s$ and all profiles in  $s'$ versus the matching distance between the profile and its correspondence in $s'$.}
        \label{fig:avg_sim_matching_value}
\end{figure*}
\subsection{Evaluation and Discussion}\label{section:data_evaluation}
In this subsection, we inspect and interpret the results of our experiments with regard to our aforementioned goals.
Therefore, we first start by giving evidence that our approach indeed provides a suitable abstraction of real users for reasoning about their privacy.

To this end, we compare the distance of matching profiles to the average distance of non-matching profiles. In particular, for each pair of profiles from the same user in subreddits $s$ and $s'$, we plot the average distance from the profile in $s$ to the non-matching profiles in $s'$ in relation to the distance to the matching profile in $s'$ in Figure~\ref{fig:avg_sim_matching_value}. The red line denotes the function $y=x$ and divides the figure into two parts: if a point lies below the line through the origin, the corresponding profiles match better than the average of the remaining profiles. Since the vast majority of datapoints is located below the line, we can conclude that profiles of the same user match better than profiles of different users.
\begin{figure*}[!t]
    \centering
        \includegraphics[width=0.8\textwidth]{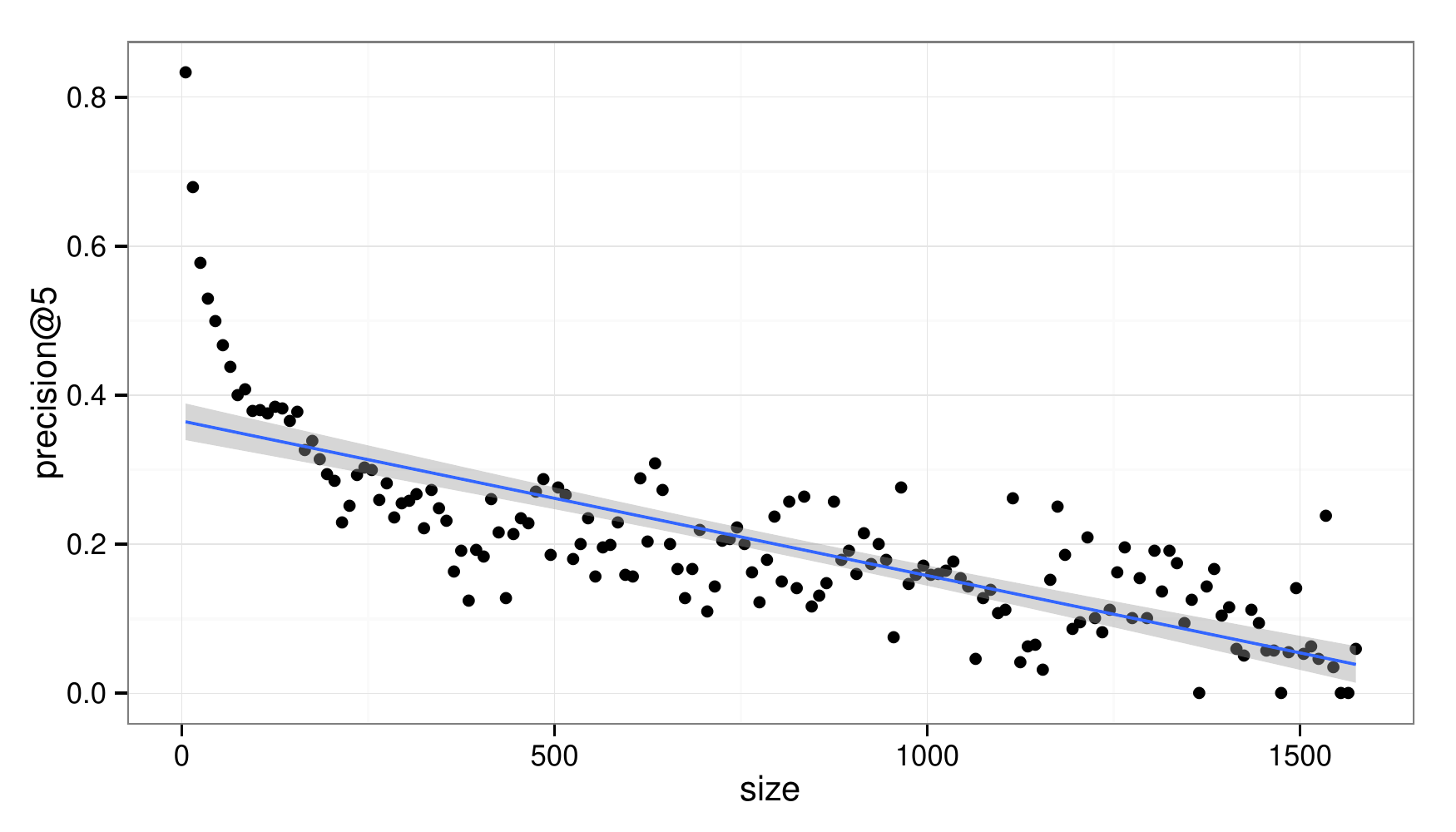}
        \caption{The anonymous subset size correlated to the precision an adversary has if considering the top $5$ profiles as matching.}
        \label{fig:precision5}
\end{figure*}

Our second goal aimed at showing that anonymous subsets indeed can be used to reason about the users' privacy.
Therefore, we investigate the chances of an adversary to find a profile of the same user within the top $k$ matches and relate its chance to the size of the profile's anonymous subset.
More precisely, given multiple target profiles with similar anonymous subset sizes, we determine the, so called, precision@$k$, i.e., the ratio of target profiles that occur in the top $k$ ranked matches (by ascending distance from the source profiles).
We relate this precision@$k$ to the anonymous subset sizes with a convergence $d$ set to the distance between the source and target profiles, and we group the anonymous subset sizes in bins of size $10$.

In our evaluation, we considered $k \in \{1,5,10,20\}$, which all yield very similar results. Exemplarily, we correlate the aforementioned measures for $k=5$ in Figure~\ref{fig:precision5}, clearly showing that an increasing anonymous subset size correlates with an increasing uncertainty -- i.e., decreasing precision -- for the adversary. 

\section{Conclusion and Future Work}\label{section:conclusion}
We presented a user-centric privacy framework for reasoning about privacy in open web settings. In our formalization, we address the essential challenges of privacy in open settings: we defined a comprehensive data model that can deal with the unstructured dissemination of heterogeneous information, and we derived the sensitivity of information from user-specified and context-sensitive privacy requirements. We showed that, in this formalization of privacy in open settings, hard security 
guarantees in the sense of Differential Privacy are impossible to achieve.     
We then instantiated the general framework to reason about the identity disclosure problem. 
The technical core of our identity disclosure model 
is the new notion of $(k,d)$-anonymity that assesses the anonymity of entities based on their similarity to other entities within the 
same community. 
We applied this instantiation to a dataset of 15 million user-generated text entries  collected from the Online Social Network Reddit and showed that our framework is suited for the assessment of linkablity threats in Online Social Networks.


As far as future work is concerned, many directions are highly promising.
First, our general framework only provides a static view on privacy in open settings. Information dissemination on the Internet, however, is, in particular, characterized by its highly dynamic nature. Extending the model presented in this paper with a suitable transition system to capture user actions might lead to powerful system for monitoring privacy risks in dynamically changing, open settings.
Second, information presented in Online Social Networks is often highly time-sensitive, e.g., shared information is often only valid for a certain period of time, and personal facts can change over time. Explicitly including timing information in our entity model will hence further increase the accuracy of the entity models derived from empirical evidence.
Finally, our privacy model is well-suited for the evaluation of protection mechanisms for very specific privacy requirements, and new such mechanisms with provable guarantees against restricted adversaries can be developed.
On the long run, we pursue the vision of providing the formal foundations for comprehensive, trustworthy privacy assessments and, ultimately, for developing user-friendly privacy assessment tools.

\bibliographystyle{plain}
\bibliography{main.bbl}

\end{document}